\documentclass[12pt]{article}

\usepackage[a4paper,top=2.5cm,bottom=2.5cm,left=2.5cm,right=2.5cm,marginparwidth=1.75cm]{geometry}
\usepackage{lineno}
\usepackage{natbib}
\usepackage{amssymb}
\usepackage{bm}
\usepackage{authblk}
\usepackage{amsmath}
\usepackage{amsthm}
\usepackage{epsfig}
\usepackage{graphicx}
\usepackage{graphics}
\usepackage{float}
\usepackage{subfigure}
\usepackage{multirow}
\usepackage{color}
\usepackage{fullpage}
\usepackage[normalem]{ulem} 
\usepackage{makeidx}
\usepackage{xspace}
\usepackage{wrapfig}
\usepackage{setspace}
\usepackage{kotex}
\usepackage{url}
\makeindex

\newtheorem{theorem}{Theorem}

\newcommand{\xb}{{\bm{x}}}
\newcommand{\Xb}{{\bm{X}}}

\newcommand{\wb}{{\bm{w}}}

\newcommand{\betab}{{\bm{\beta}}}

\newcommand{\alphab}{{\bm{\alpha}}}

\providecommand{\keywords}[1]
{
  \small	
  \textbf{\textit{Keywords: }} #1
}

\begin{document}

\title{Mixture of partially linear experts}
 
\author[1]{Yeongsan Hwang}
\author[1]{Byungtae Seo}
\author[2*]{Sangkon Oh}
\affil[1]{Department of Statistics, Sungkyunkwan University}
\affil[2*]{Department of Statistics, Ewha Womans Univerisity}
\affil[*]{Corresponding author(s). E-mail(s): sangkonoh92@gmail.com}

\date{}

\maketitle

\begin{abstract}
In the mixture of experts model, a common assumption is the linearity between a response variable and covariates. While this assumption has theoretical and computational benefits, it may lead to suboptimal estimates by overlooking potential nonlinear relationships among the variables. To address this limitation, we propose a partially linear structure that incorporates unspecified functions to capture nonlinear relationships. We establish the identifiability of the proposed model under mild conditions and introduce a practical estimation algorithm. We present the performance of our approach through numerical studies, including simulations and real data analysis.

\end{abstract}

\section{Introduction}

\cite{quandt1972new} introduced a finite mixture of regressions (FMR) for uncovering hidden latent structures in data. It assumes the existence of unobserved subgroups, each characterized by distinct regression coefficients. Since the introduction of FMR, extensive research has been conducted to enhance its performance, with contributions from \citet{Neykov_2007}, \citet{bai2012robust}, \citet{bashir2012robust}, \citet{Hunter_2012}, \citet{Yao_2014}, \citet{song2014robust}, \citet{Zeller_2015}, \citet{zeller2019finite}, \citet{ma2021semiparametric}, \citet{zarei2022robust}, and \citet{oh2024semiparametric}. 

However, because FMRs assume that the assignment of each data point to clusters is independent of the covariates \citep{hennig2000identifiablity}, FMR can be undermined with regard to the performance of regression clustering when the assumption of assignment independence is violated.
Alternatively, \cite{jacobs1991adaptive} introduced the mixture of linear experts (MoE), allowing for the assignment of each data point to depend on the covariates. \cite{nguyen2016laplace} suggested the Laplace distribution for the error distributions, while \cite{chamroukhi2016robust} and \cite{chamroukhi2017skew} used $t$ distributions and skew-$t$ distributions for errors, respectively. 
\cite{murphy2019gaussian} further extended MoE with a parsimonious structure to improve estimation efficiency. \citet{mirfarah2021mixture} introduced the use of scale mixture of normal distributions for errors within MoE. Recently,  \citet{oh2022merging} proposed a specific MoE variant, assuming that covariates follow finite Gaussian location-scale mixture distributions and that the response follows finite Gaussian scale mixture distributions.

In spite of extra flexibility for errors in these models, they assumed linear structures in each mixture component, which makes too simple to capture the hidden latent structures. In homogeneous population, \cite{engle1986semiparametric} introduced a partial linear model, comprising  a response variable $Y$ is represented as a linear combination of specific $p$-dimensional covariates $\Xb$ and an unspecified non-parametric function that includes an additional covariate $U$, as follows.
\begin{align}
\label{PLM}
      y = \boldsymbol{x}^{\top}\boldsymbol{\beta} + g({u}) + \epsilon, 
\end{align}
where $U \subset \mathbb{R}$, $\epsilon$ is an error term with a mean zero and finite variance, and the function $g(\cdot)$ is an unknown non-parametric function. 
This model has the advantages of interpretability, stemming from its linearity, with the flexibility to capture diverse functional relationships through an unspecified function $g(\cdot)$. The differentiation between $\Xb$ and U is determined either theoretically based on established knowledge in the application field or through methods like scatter plots or statistical hypothesis testing. 
\cite{wu2017estimation} and \cite{skhosana2023novel} suggested the FMR to accommodate a partially  linear structure within a heterogeneous population. 

In this paper, we consider a novel approach that incorporates partially linear structures into MoE, utilizing unspecified functions based on kernel methods. This allows proposed model to effectively capture various relationships between the response and covariates, while latent variable is dependent on some covariates. This flexibility can significantly impact the estimation of regression coefficients and enhance clustering performance by mitigating misspecification problems arising from assumptions about the relationships between variables. In addition, we address the issue of identifiability in the proposed model to ensure the reliability of the outcomes derived from proposed approach. 

The remainder of this paper is organized as follows. Section 2 reviews MoE and introduces the proposed models, addressing the identifiability. Section 3 outlines the estimation procedure, while Section 4 deals with practical issues related to the proposed models. We present the results of simulation studies in Section 5 and apply the models to real datasets in Section 6. Finally, we provide a discussion in Section 7.

\section{Semiparametric mixture of partially linear experts}
\subsection{Mixture of linear experts}

Let $Z$ be a latent variable indicating the membership of the observations. 
MoE is a useful tool when exploring the relationship between the response variable and covariates in the presence of unobserved information about $C$ heterogeneous subpopulations by latent variable $Z$. \cite{jacobs1991adaptive} presented the conditional probability distribution of the response variable given the covariates as
\begin{align}
\label{moe}
    p(y | \boldsymbol{x})  = \sum_{c=1}^{C}p(Z=c \mid \xb)p(y \mid \xb,Z=c) = \sum^C_{c=1} \pi_c(\boldsymbol{x}) \phi(y; \beta_{0c} + \boldsymbol{x}^{\top}\boldsymbol{\beta}_c,\sigma_c^2), 
\end{align}
where $\pi_c(\cdot)$, $c=1, \ldots, C$, represents a mixing probability that depends on the given covariates, with $0 < \pi_c(\xb) < 1$ and $\sum_{c=1}^{C}\pi_c(\xb) = 1$. Additionally, $(\beta_{0c}, \betab_c^{\top})$ represents a $(p+1)$-dimensional vector for $c=1, \ldots, C$, and $\phi(\cdot ;\mu, \sigma^2) $ denotes the probability density function of the normal distribution with mean $\mu$ and variance $\sigma^2$. 

Regression clustering, the process of identifying the latent variable $Z$, holds significant importance in understanding the prediction mechanism employed by MoE. The predicted value of the response variable for new covariate $\Xb = \xb$ is determined as 
\begin{align*}
E(Y \mid \Xb = \xb) = \sum^C_{c=1} \pi_c(\boldsymbol{x}) \cdot (\beta_{0c} + \boldsymbol{x}^{\top}\boldsymbol{\beta}_c), 
\end{align*}
where $\pi_c(\xb)$ is often called as the gating network, while $(\beta_{0c} + \boldsymbol{x}^{\top}\boldsymbol{\beta}_c)$ is referred to as the expert network.
That is, the prediction structure can be understood as an ensemble model as shown in Figure \ref{moe_structure} because the predicted values are obtained by combining the outcomes of the expert networks using the gating network.
Consequently, selecting an appropriate latent variable $Z$ is a crucial aspect of the MoE model.

\begin{figure}[h]
\centering
\includegraphics[width=0.8\textwidth]{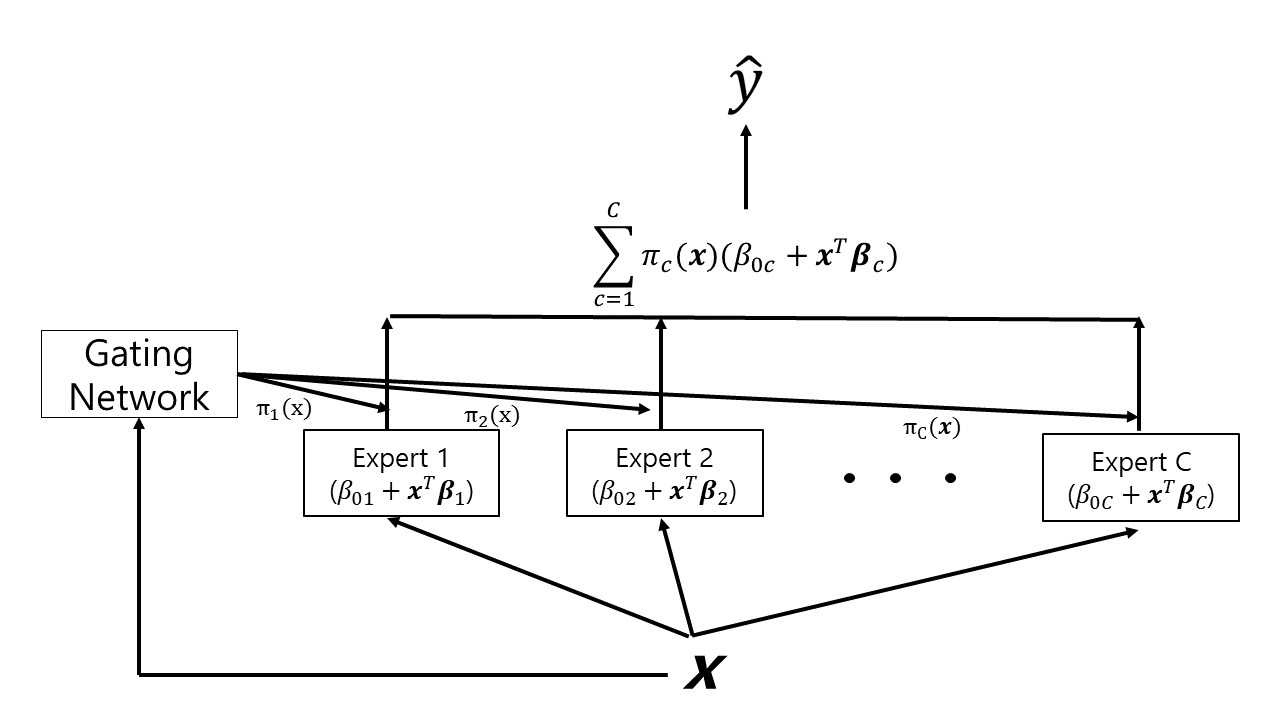}
\caption{Predicting mechanism of MoE}\label{moe_structure}
\end{figure}

MoE is applied in various fields as a machine learning model. For example, \cite{li2019modeling} used MoE to explain differences in lane-changing behavior based on driver characteristics. \cite{shen2019mixture} extended MoE to adapt to the characteristics of data for creating a translation model capable of various translation styles. Additionally, \cite{riquelme2021scaling} proposed Vision MoE, which maintains superior performance compared to existing models in image classification while significantly reducing estimation time.

\subsection{Proposed model}

In this section, we introduce a semiparametric mixture of partially linear experts (MoPLE) model. The MoPLE is constructed by considering each expert network of the MoE model as a partial linear model \eqref{PLM}, which can be defined as
\begin{equation}
\label{MoPLE}
        p(y \mid \boldsymbol{x},u) = \sum^C_{c=1} \pi_c(\boldsymbol{x}; \alpha_{0c}, \alphab_c) \phi(y;\boldsymbol{x}^{\top}\boldsymbol{\beta}_c+ g_c(u),\sigma^2_c). 
\end{equation}
Here, $\pi_c(\boldsymbol{x}; \alpha_{0c}, \alphab_c)$ is defined as $\pi_c(\boldsymbol{x}; \alpha_{0c}, \alphab_c) = \frac{\exp(\alpha_{0c} + \boldsymbol{x}^{\top}\boldsymbol{\alpha}_c)}{\sum^C_{j=1}\exp(\alpha_{0j} + \boldsymbol{x}^{\top}\boldsymbol{\alpha}_j)}$,
where $(\alpha_{0c}, \alphab_c^{\top})$ represents a $(p+1)$-dimensional vector ($c = 1,2,\ldots, C$), especially with $(\alpha_{0C}, \boldsymbol{\alpha}_C^{\top})$ being a zero vector.
When $C=1$, since $\pi_C(\boldsymbol{x}; \alpha_{0C}, \alphab_C)$ is equal to $1$, \eqref{MoPLE} simply represents a partial linear model \eqref{PLM}. If $C > 1$ and $g_c(\cdot) = 0$, \eqref{MoPLE} is equivalent to the MoE \eqref{moe}.

Identifiability is a fundamental concern when dealing with finite mixture models. \cite{hennig2000identifiablity} established that finite mixture of regressions is identifiable when the domain of $\Xb$ includes an open set in $\mathbb{R}^p$.
Additionally, \cite{huang2012mixture} demonstrated that \eqref{moe}, with unspecified $\pi_c(\xb)$ for $c=1,2,\ldots,C$, is identifiable up to a permutation of relabeling.
Furthermore, \cite{wu2017estimation} extended these findings by establishing the identifiability of the mixture of partially linear regressions, assuming that $\alphab = (\alphab_1^{\top}, \alphab_2^{\top}, \ldots, \alphab_C^{\top})^{\top}$ is a zero vector in \eqref{MoPLE}.
Building upon these results, the following theorem establishes the identifiability of model \eqref{MoPLE}.
\\
 
\begin{theorem}
    \label{thm1}
Suppose that the functions $g_c(\cdot)$, $c = 1, 2, \ldots, C$, are continuous, and the parameter vectors $(\betab_c, \sigma_c^2)$ are distinct in $\mathbb{R}^{p+1}$ for $c = 1, 2, \ldots, C$. Additionally, assume that the covariate $\Xb$ does not contain a constant, and none of its components can be a deterministic function of $U$. If the support of $\Xb$ contains an open set in $\mathbb{R}^p$, then \eqref{MoPLE} is identifiable up to a permutation of its components for almost all $(\xb^{\top}, u)^{\top} \in \mathbb{R}^{p+1}$.
\end{theorem}
\begin{proof}

In \eqref{MoPLE}, suppose that there exist $\tilde\alpha_{0k}$, $\tilde\alphab_k$, $\tilde\betab_{k}$ and $\tilde{g}_k(\cdot)$, $k = 1,2,\ldots,K$,  satisfying 
\begin{align}
\label{two models}
\sum^C_{c=1} \pi_c(\boldsymbol{x}; \alpha_{0c}, \alphab_c) \phi(y;\boldsymbol{x}^{\top}\boldsymbol{\beta}_c+ g_c(u),\sigma^2_c) = 
\sum_{k=1}^{K} \pi_k(\boldsymbol{x}; \tilde{\alpha}_{0k}, \tilde{\alphab}_k) \phi(y;\boldsymbol{x}^{\top} \tilde{\boldsymbol{\beta}}_k+ \tilde{g}_k(u),\tilde{\sigma}^2_k),
\end{align}
where ($\tilde{\betab}_k, \tilde{\sigma}_k^2$), $k = 1,2,\ldots, K$, are distinct.
Consider the set $\{\xb \in \mathbb{R}^p : \xb^{\top} \betab_{c_1} + g_{c_1}(u) = \xb^{\top} \betab_{c_2} + g_{c_2}(u) \}$ for any $\betab_{c_1}$ and $\betab_{c_2}$ ($c_1$, $c_2$ $\in$ ${1,2,\ldots,C }$ ), where $\betab_{c_1} \neq \betab_{c_2}$ and $\sigma^2_{c_1} = \sigma^2_{c_2}$, for a given $U = u$. This set represents a $(p-1)$-dimensional hyperplane in $\mathbb{R}^p$. For any pair of $\betab_{c_1}$ and $\betab_{c_2}$ with $\betab_{c_1} \neq \betab_{c_2}$ and $\sigma^2_{c_1} = \sigma^2_{c_2}$, the union of a finite number of such hyperplanes, where $(\xb^{\top} \betab_{c_1}, \sigma^2_{c_1}) = (\xb^{\top} \betab_{c_2}, \sigma^2_{c_2})$, has a zero Lebesgue measure in $\mathbb{R}^p$. This fact remains true for the finite number of sets $\{\xb \in \mathbb{R}^p : \xb^{\top} \tilde{\betab}_{k_1} + \tilde{g}_{k_1}(u) = \xb^{\top} \tilde{\betab}_{k_2} + \tilde{g}_{k_2}(u) \}$ for any $\tilde{\betab}_{k_1}$ and $\tilde{\betab}_{k_2}$ ($k_1$, $k_2$ $\in$ $\{1,2,\ldots,K \}$ ), where $\tilde{\betab}_{k_1} \neq \tilde{\betab}_{k_2}$ and $\tilde{\sigma}^2_{k_1} = \tilde{\sigma}^2_{k_2}$ for given $U = u$.

From Lemma 1 of \cite{huang2012mixture}, it can be established that \eqref{two models} is identifiable when conditioned on $\bm{w} = (\xb^{\top}, u)^{\top}$, under the condition that both sets of $(\boldsymbol{x}^{\top}\boldsymbol{\beta}_c, g_c(u))$ for $c = 1, 2, \ldots, C$ and $(\boldsymbol{x}^{\top} \tilde{\boldsymbol{\beta}}_{k},  \tilde{g}_{k} (u))$ for $k = 1,2,\ldots,K$ are distinct. 
That is, if $\bm{w}$ is given, we obtain $C = K$, and there exists a permutation $\tau_\wb = \{ \tau_\wb{(1)}, \tau_\wb{(2)}, \ldots, \tau_\wb{(C)} \}$ among the finite number of possible permutations of $\{1,2,\ldots,C \}$ such that  
\begin{align*}
\pi_c(\boldsymbol{x}; \alpha_{0c}, \alphab_c) = \pi_{\tau_\wb{(c)}}(\boldsymbol{x}; \tilde{\alpha}_{0{\tau_\wb{(c)}}}, \tilde{\alphab}_{\tau_\wb{(c)}}), \ \ \boldsymbol{x}^{\top}\boldsymbol{\beta}_c+ g_c(u) = \boldsymbol{x}^{\top} \tilde{\boldsymbol{\beta}}_{\tau_\wb{(c)}}+ \tilde{g}_{\tau_\wb{(c)}}(u), \ \ \sigma^2_c = \tilde{\sigma}^2_{\tau_\wb{(c)}} 
\end{align*}
where $c = 1,2,\ldots,C$.

Now, let us consider any permutation $\tau = \{ \tau{(1)}, \tau{(2)}, \ldots, \tau{(C)} \}$ that satisfies 
\begin{align}
\label{ident}
\boldsymbol{x}^{\top}\boldsymbol{\beta}_c+ g_c(u) = \boldsymbol{x}^{\top} \tilde{\boldsymbol{\beta}}_{\tau{(c)}}+ \tilde{g}_{\tau{(c)}}(u), \ \ \sigma^2_c = \tilde{\sigma}^2_{\tau{(c)}}, \ \ c = 1,2,\ldots,C,
\end{align}
for some $\wb$,
and verify that $\tau_\wb$ has to be unique $\tau$.
Suppose that $\boldsymbol{\beta}_c \neq \tilde{\boldsymbol{\beta}}_{\tau{(c)}}$ and $g_c(u) \neq \tilde{g}_{\tau{(c)}}(u)$. This contradicts to the assumption that $\Xb$ cannot be a deterministic function of $U$. When $\boldsymbol{\beta}_c \neq \tilde{\boldsymbol{\beta}}_{\tau{(c)}}$ and $g_c(u) = \tilde{g}_{\tau{(c)}}(u)$, the set $\{\xb \in \mathbb{R}^p : \xb^{\top} \betab_{c} = \xb^{\top} \tilde{\boldsymbol{\beta}}_{\tau{(c)}} \}$ has zero Lebesgue measure since it is a $(p-1)$ dimensional hyperplane in $\mathbb{R}^p$. Because $\boldsymbol{\beta}_c = \tilde{\boldsymbol{\beta}}_{\tau{(c)}}$ indicates $g_c(u) = \tilde{g}_{\tau{(c)}}(u)$, we obtain that 
\begin{align*}
\boldsymbol{\beta}_c  = \tilde{\boldsymbol{\beta}}_{\tau{(c)}}, \ \ g_c(u) =  \tilde{g}_{\tau{(c)}}(u)
\end{align*}
for $c = 1,2,\ldots,C$.
Since the parameter sets $(\betab_c, \sigma_c^2)$ and $(\tilde{\boldsymbol{\beta}}_k, \tilde{\sigma}^2_k)$ for $c, k \in \{1,2,\ldots,C \}$ are distinct, the permutation $\tau$, satisfying \eqref{ident} on a subset of the support of $\wb$ with nonzero Lebesgue measure, is unique.

Because $\pi_c(\cdot)$ and $\pi_{\tau{(c)}}(\cdot)$ are continuous and one to one function, it follows that $\alpha_{0c} + \xb^{\top} \alphab_c = \tilde{\alpha}_{0{\tau{(c)}}} + \xb^{\top} \tilde{\alphab}_{\tau{(c)}}$ for $c = 1,2,\ldots,C$. Moreover, as $\Xb$ cannot be a constant, $\alpha_{0c} = \tilde{\alpha}_{0{\tau{(c)}}}$ must be hold. Consequently, this indicates $\alphab_c = \tilde{\alphab}_{\tau{(c)}}$ , except for the set $\{\xb \in \mathbb{R}^p : \alpha_{0c} + \xb^{\top} \alphab_{c} = \tilde{\alpha}_{0{\tau{(c)}}} + \xb^{\top} \tilde{\alphab}_{\tau{(c)}} \}$, which has a zero Lebesgue measure in $\mathbb{R}^p$, for $c = 1,2,\ldots,C$.
Therefore, we can conclude that \eqref{MoPLE} is identifiable up to a permutation of its components. 
\end{proof}

\section{Estimation}

When considering the observed data $\{(y_i,\boldsymbol{x}_i,u_i)\}^n_{i=1}$, the log-likelihood function is defined as
\begin{equation}
    \ell(\boldsymbol{\Theta}, \boldsymbol{g}) = \sum^n_{i=1}\log \Bigg [\sum^C_{c=1} \pi_c(\xb) \phi\{y_i ; \boldsymbol{x}_i^{\top}\boldsymbol{\beta}_c + g_c(u_i), \sigma^2_c\} \Bigg],  \label{likelihood}
\end{equation}
where $\boldsymbol{\Theta}$ is the set of all parameters and $\boldsymbol{g} = (g_1(\cdot), \ldots, g_C(\cdot))^{\top}$.
To find $\hat{\boldsymbol{\Theta}}$ and $\hat{\boldsymbol{g}}$ that maximize equation \eqref{likelihood}, we propose the Expectation Conditional Maximization (ECM) algorithm \citep{meng1993maximum} using the profile likelihood method. The latent indicator variable $Z_{ic}$ ($c=1,\ldots,C$), which indicates to which latent cluster the observed values belong, and the complete log-likelihood function are respectively defined as 
\begin{equation*}
Z_{ic} =
\begin{cases}
1, & \text{if the $i$-th observation belongs to the $c$-th latent cluster} \\
0, &  \text{otherwise}
\end{cases}
\end{equation*}
and
\begin{equation*}
    \ell_c(\boldsymbol{\Theta}, \boldsymbol{g}) = \sum^n_{i=1}\sum^C_{c=1}Z_{ic} \log \Bigg[\pi_c(\xb) \phi\{y_i|\boldsymbol{x}^{\top}_i\boldsymbol{\beta}_c + g_c(u_i),\sigma^2_c\} \Bigg]. 
\end{equation*}

In the E-step for the $(t+1)$th iteration of the ECM algorithm, $t=0,1,\ldots$, we obtain $Q(\boldsymbol{\Theta}^{(t)}, \boldsymbol{g}^{(t)}) = E[\ell_c(\boldsymbol{\Theta}, \boldsymbol{g}) | \boldsymbol{\Theta}^{(t)}, \boldsymbol{g}^{(t)}]$ using the posterior probability ${z}_{ic}^{(t+1)}$ given ${\boldsymbol{\Theta}}^{(t)}$ and ${\boldsymbol{g}}^{(t)}$, which is represented as
\begin{align*}
       {z}_{ic}^{(t+1)} = E(Z_{ic} |\boldsymbol{x}_i,y_i,\boldsymbol{\Theta}^{(t)}, \boldsymbol{g}^{(t)})
       = \frac{\pi_c^{(t)}(\xb)\phi\{y_i;\boldsymbol{x}^T_i\boldsymbol{\beta}_c^{(t)} + g^{(t)}_c(u_i),{\sigma_c^{2}}^{(t)}\}}{\sum^C_{j=1} \pi_j^{(t)}(\xb) \phi\{y_i;\boldsymbol{x}^{\top}_i\boldsymbol{\beta}_j^{(t)}+ {g_j^{{(t)}}(u_i)} ,{\sigma_j^2}^{(t)}\}}.  
\end{align*}
While keeping $\boldsymbol{\Theta}^{(t)}$ $(c = 1,2,\ldots,C)$ fixed, CM-step 1 involves updating ${\boldsymbol{g}}^{(t)}$ to ${\boldsymbol{g}}^{(t+1)}$ that maximizes the following local likelihood:
\begin{equation*}
    \ell_h(\boldsymbol{g})= \sum^n_{i=1}\sum^C_{c=1}{z}^{(t+1)}_{ic} \Bigg[\log\phi\{y_i;\boldsymbol{x}^T_i\boldsymbol{\beta}_c^{(t)} + g_c(u_j),{\sigma_j^2}^{(t)}\}\Bigg]K_h(u_i-u_j), 
\end{equation*}
where $j \in \{1,2, \ldots, n \}$, and $K_h({u}_i-{u}_j)$ represents the kernel weighting function with bandwidth $h$. 
Consequently, ${g}_c^{(t+1)}(u_j)$ can be calculated as
\begin{equation*}
    g^{(t+1)}_c(u_j) = \frac{\sum^n_{i=1} {z}^{(t+1)}_{ic} (y_i-\boldsymbol{x}^{\top}_i\boldsymbol{\beta}^{(t)}_c) K_h(u_i-u_j)}{\sum^n_{i=1}{z}^{(t+1)}_{ic}K_h(u_i-u_j)}. 
\end{equation*}
In CM-step 2, after fixing ${g}_c^{(t+1)}(u_j)$, we can determine $\boldsymbol{\Theta}^{(t+1)}$ as follows.
\begin{align*}
    \boldsymbol{\alpha}_c^{(t+1)} &= \boldsymbol{\alpha}_c^{(t)} - \Bigg[\frac{\partial^2 Q(\boldsymbol{\Theta}^{(t)}, \boldsymbol{g}^{(t+1)})}{\partial\boldsymbol{\alpha}_c\partial\boldsymbol{\alpha}_c^{\top}} \Bigg]^{-1}  \Bigg[
    \frac{\partial Q(\boldsymbol{\Theta}^{(t)}, \boldsymbol{g}^{(t+1)})}{\partial\boldsymbol{\alpha}_c} \Bigg],
\end{align*}
\begin{align*}
    \boldsymbol{\beta}^{(t+1)}_c = (\Tilde{\boldsymbol{X}}^{\top}{\boldsymbol{Z}}_c^{(t+1)}\Tilde{\boldsymbol{X}})^{-1}\Tilde{\boldsymbol{X}}^{\top}{\boldsymbol{Z}}_c^{(t+1)}\Tilde{\boldsymbol{y}}, 
\end{align*}
\begin{align*}
{\sigma_c^{2}}^{(t+1)} = \frac{\sum^n_{i=1}{z}_{ic}^{(t+1)}(y_i-\boldsymbol{x}_i\boldsymbol{\beta}^{(t+1)}_c-g^{(t+1)}_c(u_i))^2 }{\sum^n_{i=1}{z}_{ic}^{(t+1)}}. 
\end{align*}
Here, $\Tilde{\boldsymbol{X}} = (\boldsymbol{I}-\boldsymbol{S})\boldsymbol{X}$, $\Tilde{\boldsymbol{y}} = (\boldsymbol{I}-\boldsymbol{S})\boldsymbol{y}$, ${\boldsymbol{Z}}_c^{(t+1)}$ is a diagonal matrix with diagonal elements $z_{ic}^{(t+1)}$, $\boldsymbol{I}$ is a $n \times n$ identity matrix, and $\boldsymbol{S}$ is a $n \times n$ matrix with elements defined as
\begin{align*}
\boldsymbol{S}_{ij} = \frac{{z}^{(t+1)}_{ic}K_h(u_i-u_j)}{\sum^n_{i=1}{z}^{(t+1)}_{ic}K_h(u_i-u_j)}.
\end{align*}

\section{Practical issues} 

In practice, it is recommend to explore multiple initial values when employing the ECM algorithm, as the mixture likelihood inherently exhibits multimodality. To acquire appropriate initial values, we utilize the mixture of linear experts approach as proposed by \cite{jacobs1991adaptive} for parameters such as $\alpha_{0c}$, $\alphab_c$, $\betab_c$, $g_c(u)$, and $\sigma_c^2$, where $c = 1, 2, \ldots, C$.
Specifically, we set $g_c(u)$ as $\beta_{0c}$ in \eqref{moe} when employing the mixture of linear experts, where $c = 1, 2, \ldots, C$. Multiple initial values are then selected by repeating the process of generating initial values and choosing the ones with the highest likelihood. In this study, we repeat this process 10 times to ensure the acquisition of suitable initial values.

Furthermore, it is crucial to employ suitable methods for determining the optimal number of mixture components. In this paper, we utilized the Bayesian information criterion (BIC; \citealt*{schwarz1978estimating}) obtained as $ -2\mathcal{\ell} + \log(n)\, \times\, \mathit{df}$, where $\mathcal{\ell}$ is the log-likelihood function and $\mathit{df}$ is degree of freedoms, to select the number of components.
However, directly applying the BIC to the proposed model is challenging due to the complexity of calculating degrees of freedom, particularly in the presence of non-parametric functions. Therefore, we adopt a modified approach for determining degrees of freedom, inspired by \cite{wu2017estimation}, as follows.
\begin{equation*}
    \mathit{df}\,=\, C \times \tau_{K}h^{-1}\vert\Omega\vert \left\{K(0) - \frac{1}{2}\int K^2(t)dt\right\} + (2C-1)(p+1), 
\end{equation*}
where $\Omega$ represents the support of the non-parametric component covariates and
\begin{equation*}
    \tau_{K} = \frac{K(0) - 0.5\int K^2(t)dt}{\int\{K(t) - 0.5K(t)\}^2dt}. 
\end{equation*}
{Given that the degrees of freedom depends on the bandwidth, we chose the bandwidth associated with the lowest BIC among the candidates.}

\section{Simulaton studies}
In this section, we present simulation results demonstrating the performance of the proposed method compared to other estimation methods under various cases. Specifically, we consider the following methods for each simulated sample: 
\begin{enumerate}
\item MoE: Mixture of linear experts  
\item FMPLR: Finite mixture of partially linear regressions. 
\item MoPLE: Mixture of partially linear experts. 
\end{enumerate}
FMPLR was introduced by \cite{wu2017estimation}, where it is assumed that all $\alphab = (\alphab_1, \alphab_2, \ldots, \alphab_C)$ to be zero vectors. We utilize the $\texttt{MoEClust}$ in \textsf{R} package \citep{MoEClust} for MoE, while we implement our \textsf{R} program for FMPLR and MoPLE. 

We conduct three simulation scenarios, each comprising two mixture components as detailed in Table \ref{tab:simulations}. In each of these experiments, we assume that the covariates $X$ and $U$ are independent random variables following a standard uniform distribution. In the first experiment, we assume a linear relationship between $Y$ and $(X, U)$ within each mixture component, with the probability of observations belonging to latent clusters dependent on $X$.
In the second experiment, we introduce partially linear relationships between $Y$ and $(X, U)$ while keeping the probability of observations belonging to latent clusters independent of $X$.
In the third experiment, we also consider partially linear relationships, but it features the probability of observations belonging to latent clusters as dependent on $X$. {Hence, we can expect that MoE, FMPLR and MoPLE represent efficient methods for Case $\uppercase\expandafter{\romannumeral1}$, Case $\uppercase\expandafter{\romannumeral2}$, and Case $\uppercase\expandafter{\romannumeral3}$, respectively.}

\begin{table}[ht]
\caption{True parameters for each simulation scenarios}
\label{tab:simulations} 
	\centering
	\begin{tabular}{ c | c c | c c c  | c c c } 
       \hline \hline\noalign{\smallskip}
		\multirow{2}{*}{Scenarios} & \multicolumn{2}{c|}{Gating Network} & \multicolumn{3}{c|}{Component $1$}  & \multicolumn{3}{c}{Component $2$} \\
  & $\alpha_{01}$& $\alpha_{11}$  & $\beta_{1}$ & $g_1(u)$ & 
  $\sigma_1^2$ & $\beta_{2}$ & $g_2(u)$ & $ \sigma_2^2$   \\ 	   \hline
        Case $\uppercase\expandafter{\romannumeral1}$  &  -0.5  & 2 & -3 & -3u & 0.5 & 3 & 3u & 0.25   \\
       Case $\uppercase\expandafter{\romannumeral2}$  &  0  & 0 & -3 & $2u^2$ & 0.5 & 3 & $2\cos({\pi u})^2$ & 0.25   \\
       Case $\uppercase\expandafter{\romannumeral3}$  &  -0.5  & 2 & -3 & $2u^2$ & 0.5 & 3 & $2\cos({\pi u})^2$ & 0.25   \\
\noalign{\smallskip}\hline\noalign{\smallskip}
	\end{tabular}
\end{table}

\begin{table}[p]
\caption{Performance of each method for regression coefficients in Case $\uppercase\expandafter{\romannumeral1}$ (Boldfaced numbers indicate the best in each criterion)}
\label{tab:case1} 
	\centering
	\begin{tabular}{ c  c  c   c  c  c c c c} 
        \hline  \hline \noalign{\smallskip}
	\multirow{2}{*}{Method}	&
	\multirow{2}{*}{$n$}	&
	$\beta_1$ & $\beta_2$ & ${g}_1(\cdot)$ & ${g}_2(\cdot)$ & \multirow{2}{*}{ARI} & \multirow{2}{*}{AMI}  \\
 & & MSE  (bias) & MSE  (bias) & $\text{MAE}$  & $\text{MAE}$   & & \\ \hline
    \multirow{3}{*}{MoE} & $250$ &  $\boldsymbol{0.045}$  (0.016) & $\boldsymbol{0.037}$  (0.005) & $\boldsymbol{0.087}$ & $\boldsymbol{0.077}$ &  $\boldsymbol{0.961}$ & $\boldsymbol{0.923}$\\
                            & $500$ &  $\boldsymbol{0.024}$  (0.010) & $\boldsymbol{0.017}$  (-0.005) & $\boldsymbol{0.059}$ & $\boldsymbol{0.052}$ &  $\boldsymbol{0.962}$ & $\boldsymbol{0.922}$ \\
                            & $1000$ & $\boldsymbol{0.011}$  (-0.003) & $\boldsymbol{0.009}$  (-0.005) & $\boldsymbol{0.042}$ & $\boldsymbol{0.036}$ &  $\boldsymbol{0.963}$ & $\boldsymbol{0.923}$ \\
\hline
    \multirow{3}{*}{FMPLR} & $250$ &  0.049  (-0.033) & 0.040  (-0.023) & 0.159 & 0.130 & 0.952 & 0.908 \\
                            & $500$ &  0.026  (-0.004) & 0.019  (-0.034) & 0.113 & 0.093 & 0.954 & 0.908 \\
                            & $1000$ & 0.014  (-0.051) & 0.010  (-0.032) & 0.084 & 0.070 & 0.955 & 0.910  \\
\hline
    \multirow{3}{*}{MoPLE} & $250$ &  0.047  (0.014) & 0.040  (0.006) & 0.154 & 0.127 & 0.960 & 0.920\\
                            & $500$ &  0.024  (0.011) & 0.018  (-0.006) & 0.110 & 0.089 & 0.961 & 0.921 \\
                            & $1000$ & 0.011  (-0.001) & 0.051  (-0.016) & 0.082 & 0.081 & 0.961 & 0.920  \\
                            \hline
        \noalign{\smallskip}\hline\noalign{\smallskip}
	\end{tabular}
\end{table}

\begin{table}[p]
\caption{Performance of each method for regression coefficients in Case $\uppercase\expandafter{\romannumeral2}$ (Boldfaced numbers indicate the best in each criterion)}
\label{tab:case2} 
	\centering
	\begin{tabular}{ c  c  c   c  c  c c c c } 
        \hline  \hline \noalign{\smallskip}
	\multirow{2}{*}{Method}	&
	\multirow{2}{*}{$n$}	&
	$\beta_1$ & $\beta_2$ & ${g}_1(\cdot)$ & ${g}_2(\cdot)$  & \multirow{2}{*}{ARI} & \multirow{2}{*}{AMI}\\
 & & MSE  (bias) & MSE  (bias) & $\text{MAE}$ & $\text{MAE}$   & & \\ \hline
    \multirow{3}{*}{MoE} & $250$ &  0.077  (-0.062) & 0.120  (-0.036) & 0.362 & 1.056 & 0.652 & 0.562 \\
                            & $500$ &  0.041  (-0.079) & 0.056  (-0.033) & 0.361 & 1.063 & 0.657 & 0.562 \\
                            & $1000$ & 0.019  (-0.053) & 0.030  (-0.033) & 0.356 & 1.063 & 0.664 & 0.565 \\
\hline
    \multirow{3}{*}{FMPLR} & $250$ &  $\boldsymbol{0.069}$  (0.014) & $\boldsymbol{0.035}$  (0.015) & $\boldsymbol{0.169}$ & 0.231 & $\boldsymbol{0.737}$ & 0.639 \\
                            & $500$ &  0.079  (0.026) & 0.043  (0.005) & 0.126 & 0.204 & 0.741 & 0.643\\
                            & $1000$ & 0.033  (0.040) & $\boldsymbol{0.009}$  (0.001) & $\boldsymbol{0.095}$ & 0.160 & 0.748 & 0.649 \\
\hline
    \multirow{3}{*}{MoPLE} & $250$ &  ${0.071}$  (0.014) & ${0.035}$  (0.012) & ${0.171}$ & $\boldsymbol{0.214}$ & 0.734 & $\boldsymbol{0.640}$ \\
                            & $500$ &  $\boldsymbol{0.035}$  (0.006) & $\boldsymbol{0.018}$  (0.010) & $\boldsymbol{0.125}$ & $\boldsymbol{0.171}$ & $\boldsymbol{0.744}$ & $\boldsymbol{0.646}$ \\
                            & $1000$ & $\boldsymbol{0.022}$  (0.029) & ${0.031}$  (-0.013) & ${0.101}$ & $\boldsymbol{0.131}$  & $\boldsymbol{0.750}$ & $\boldsymbol{0.651}$ \\
                                                        \hline
        \noalign{\smallskip}\hline\noalign{\smallskip}
	\end{tabular}
\end{table}

\begin{table}[p]
\caption{Performance of each method for regression coefficients in Case $\uppercase\expandafter{\romannumeral3}$ (Boldfaced numbers indicate the best in each criterion)}
\label{tab:case3} 
	\centering
	\begin{tabular}{ c  c  c   c  c  c c c c } 
        \hline  \hline \noalign{\smallskip}
	\multirow{2}{*}{Method}	&
	\multirow{2}{*}{$n$}	&
	$\beta_1$ & $\beta_2$ & ${g}_1(\cdot)$ & ${g}_2(\cdot)$ & \multirow{2}{*}{ARI} & \multirow{2}{*}{AMI}\\
 & & MSE  (bias) & MSE  (bias) & $\text{MAE}$  & $\text{MAE}$  & & \\ \hline
    \multirow{3}{*}{MoE} & $250$ &  $\boldsymbol{0.062}$ (-0.074) & 0.230  (-0.169) & 0.361 & 1.100 & 0.641 & 0.529 \\
                            & $500$ &  $\boldsymbol{0.034}$  (-0.079) & 0.127  (-0.175) & 0.357 & 1.074 & 0.645 & 0.529\\
                            & $1000$ & 0.020  (-0.076) & 0.087  (-0.207) & 0.348 & 1.0578 & 0.652 & 0.533 \\
\hline
    \multirow{3}{*}{FMPLR} & $250$ &  0.078  (-0.118) & 0.215  (-0.170) & 0.182 & 0.270 & 0.661 & 0.556 \\
                            & $500$ &  0.044  (-0.132) & 0.075  (-0.130) & 0.146 & 0.203 & 0.671 & 0.562\\
                            & $1000$ & 0.036  (-0.123) & 0.075  (-0.145) & 0.125 & 0.193 & 0.675 & 0.566  \\
\hline
    \multirow{3}{*}{MoPLE} & $250$ &  ${0.066}$  (0.038) & $\boldsymbol{0.064}$  (-0.020) & $\boldsymbol{0.172}$ & $\boldsymbol{0.245}$ & $\boldsymbol{0.743}$ & $\boldsymbol{0.638}$ \\
                            & $500$ &  ${0.038}$  (0.038) & $\boldsymbol{0.086}$  (-0.039) & $\boldsymbol{0.123}$ & $\boldsymbol{0.200}$ & $\boldsymbol{0.748}$ & $\boldsymbol{0.641}$ \\
                            & $1000$ & $\boldsymbol{0.017}$  (0.038) & $\boldsymbol{0.076}$  (-0.045) & $\boldsymbol{0.094}$ & $\boldsymbol{0.161}$ & $\boldsymbol{0.771}$ & $\boldsymbol{0.667}$ \\
                                                                                    \hline
        \noalign{\smallskip}\hline\noalign{\smallskip}
	\end{tabular}
\end{table}

The performance of each method is evaluated by calculating the bias as $\frac{1}{r} \sum_{j = 1}^{r} (\hat{\beta}_{c(j)} - {\beta}_{c})$ and mean square error (MSE) as $\frac{1}{r} \sum_{j = 1}^{r} (\hat{\beta}_{c(j)} - {\beta}_{c})^2$, where $\beta_{c}$ and $\hat{\beta}_{c(j)}$ are the true regression coefficient in $c$th expert network and the estimate of the $\beta_{c}$ from the $j$th sample for $c=1,2$ and $j = 1,2,\ldots, r$, respectively, for every regression parameter across a total of $r = 400$ replicated samples, with sample sizes of $n=$250, 500 and 1000. 
To assess the quality of the estimated nonparametric function $\hat{\boldsymbol{g}} = (\hat{g}_1(\cdot), \hat{g}_2(\cdot))$ for $\boldsymbol{g} = ({g}_1(\cdot), g_2(\cdot))$, we utilize the mean absolute error (MAE) defined as 
\begin{align*}
    \text{MAE} = D^{-1} \sum^D_{d=1} |\hat{g}_c(u_d) - g_c(u_d) |,
\end{align*}
where $c = 1,2,\ldots, C$.
We chose $\{u_d, d=1,\ldots,D\}$ as grid points evenly distributed within the range of the covariate $u$, with $D$ set to 100. We employ the Epanechnikov kernel function and determine regression clusters for observations using the maximum a posteriori. To assess the clustering performance, the Adjusted Rand Index (ARI, Hubert and Arabie, 1985) and Adjusted Mutual Information (AMI, Vinh et al., 2009) are computed. Note that smaller values of bias, MSE and $\text{MAE}$ indicate better performance, while larger values of ARI and AMI signify better performance.

In Case $\uppercase\expandafter{\romannumeral1}$, MoE exhibits the best performance across all criteria, while MoPLE ranks second in terms of clustering performance. 
In Case $\uppercase\expandafter{\romannumeral2}$, MoPLE performs the best in terms of ARI and AMI, while FMPLR and MoPLE are competitive with regard to the estimating parameters. 
In Case $\uppercase\expandafter{\romannumeral3}$, MoPLE demonstrates the best with regard to almost all criteria compared to the other methods. Overall, MoPLE demonstrates competitive performance, ranking either as the best or the second best method across all cases.

\section{Real data analysis}

\subsection{Prestige dataset}

For the first real data analysis, we consider the Prestige dataset, which is available in the car package in \textsf{R}. It comprises 102 observations with the variable such as Prestige, indicating occupational prestige from a mid-1960s social survey, Education, representing the average years of education for workers in 1971, Income, denoting the standardized average income of workers in 1971, and Occupational types, specifying occupational categories like professional, white-collar, and blue-collar occupations.
In this study, we model the response variable $Y$ as Prestige, where $X$ represents Education, and $U$ represents Income. Additionally, we assume that the latent variable is associated with Occupational types. 

Table \ref{tab:prestige_BIC} displays the BIC values obtained by each method for the Prestige dataset. MoPLE correctly selects the expected number of components, while MoE and FMPLR yield fewer clusters than expected.
The clustering performance of each method is summarized in Table \ref{tab:prestige_clustering}. MoE performs the best in terms of ARI, whereas MoPLE excels in terms of AMI. As a result, MoPLE is considered the best method since it not only produces the expected number of clusters but also delivers competitive clustering performance. MoE is the second-best method, despite not selecting the expected number of clusters. This suggests that occupational types are dependent on education, and there are nonlinear relationships between prestige and income, at least within one component.

\begin{table}[ht]
\caption{BIC values for each method in prestige dataset (Boldfaced numbers indicate the smallest value in each criterion)}
\label{tab:prestige_BIC} 
	\centering
	\begin{tabular}{ c | r r  r    } 
        \hline \hline\noalign{\smallskip}
		Number of clusters & MoE & FMPLR & MoPLE   \\ 	   \hline
        {1} &  724.22 & 947.23 & 947.23    \\
       {2} & $\boldsymbol{718.24}$ & $\boldsymbol{864.11}$ & 852.19   \\
      {3} & 735.49 & 951.21 & $\boldsymbol{823.31}$   \\
      {4} & 736.40 & 1042.96 & 1126.14  \\
      {5} & 763.94 & 1633.26 & 1186.94   \\
\noalign{\smallskip}\hline\noalign{\smallskip}
	\end{tabular}
\end{table}

\begin{table}[ht]
\caption{Clustering performance for each method in prestige dataset (Boldfaced numbers indicate the largest value in each criterion)}
\label{tab:prestige_clustering} 
	\centering
	\begin{tabular}{ c | r r  r    } 
       \hline \hline\noalign{\smallskip}
		Index & MoE & FMPLR & MoPLE   \\ 	   \hline
        ARI &  $\boldsymbol{0.5096}$ & 0.0597 & 0.4779    \\
       AMI & 0.4012 & 0.0725 & $\boldsymbol{0.4506}$   \\
  \noalign{\smallskip}\hline\noalign{\smallskip}
	\end{tabular}
\end{table}

Based on the findings from MoPLE, the clusters denoted as $1$, $2$, and $3$ correspond to professional, white-collar, and blue-collar occupations, respectively. The estimated coefficients for the Education in Class 1, 2, and 3 are 2.331, 5.446, and 2.547, respectively. 
This suggests that the impact of the education on the prestige is most pronounced in white-collar.
{ Figure \ref{prestige} illustrates the estimated $g_c(u)$ for each cluster, where $c = 1,2,3$. We note a nonlinear association between prestige and income within cluster $1$, whereas clusters $2$ and $3$ exhibit a positive relationship between prestige and income, indicating an increasing trend.}

\begin{figure}[] 
\centering
\subfigure[Cluster 1]{
\includegraphics[width=0.3\linewidth, height=0.5\linewidth]{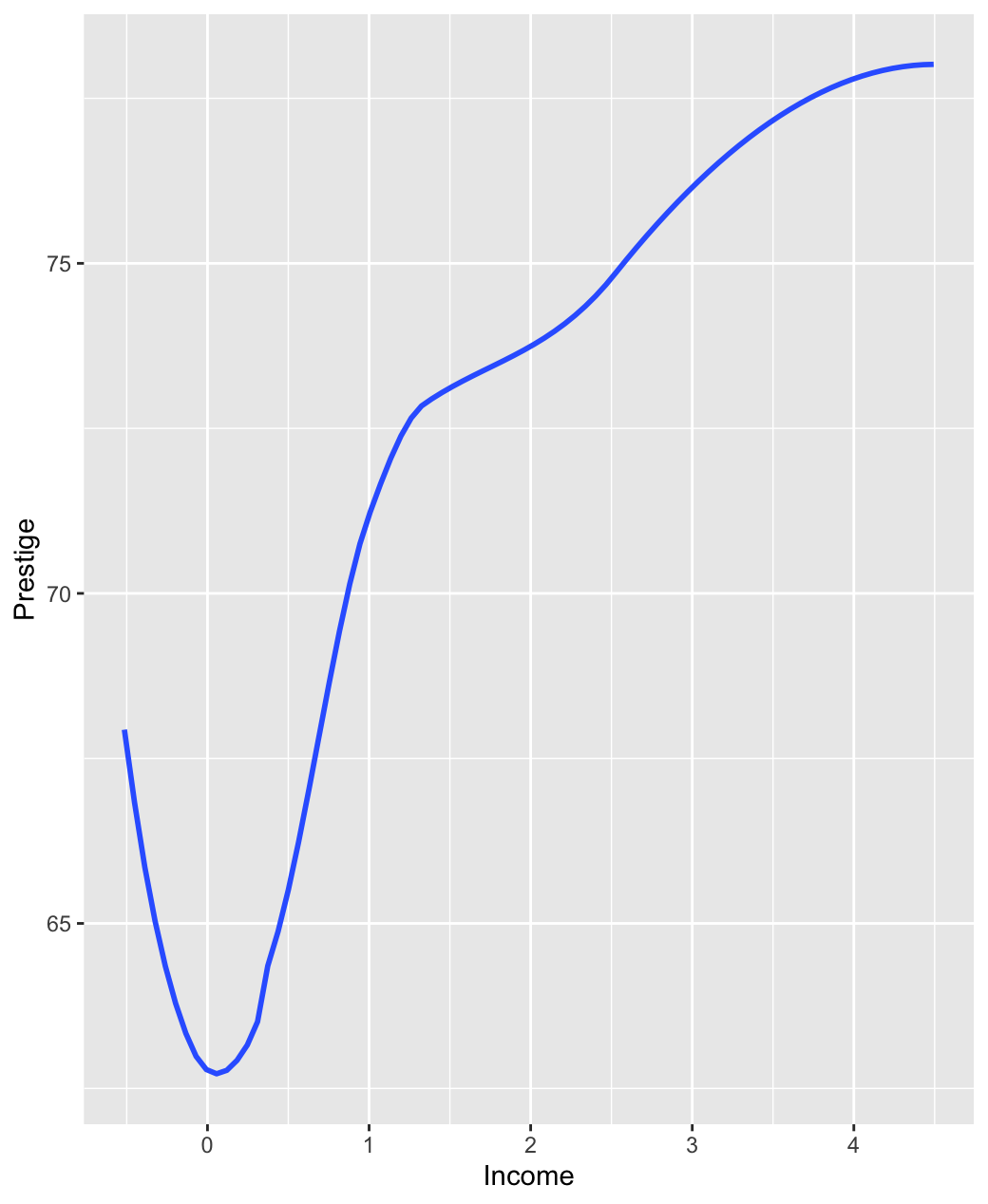}
}
\centering
\subfigure[Cluster 2]{
\includegraphics[width=0.3\linewidth, height=0.5\linewidth]{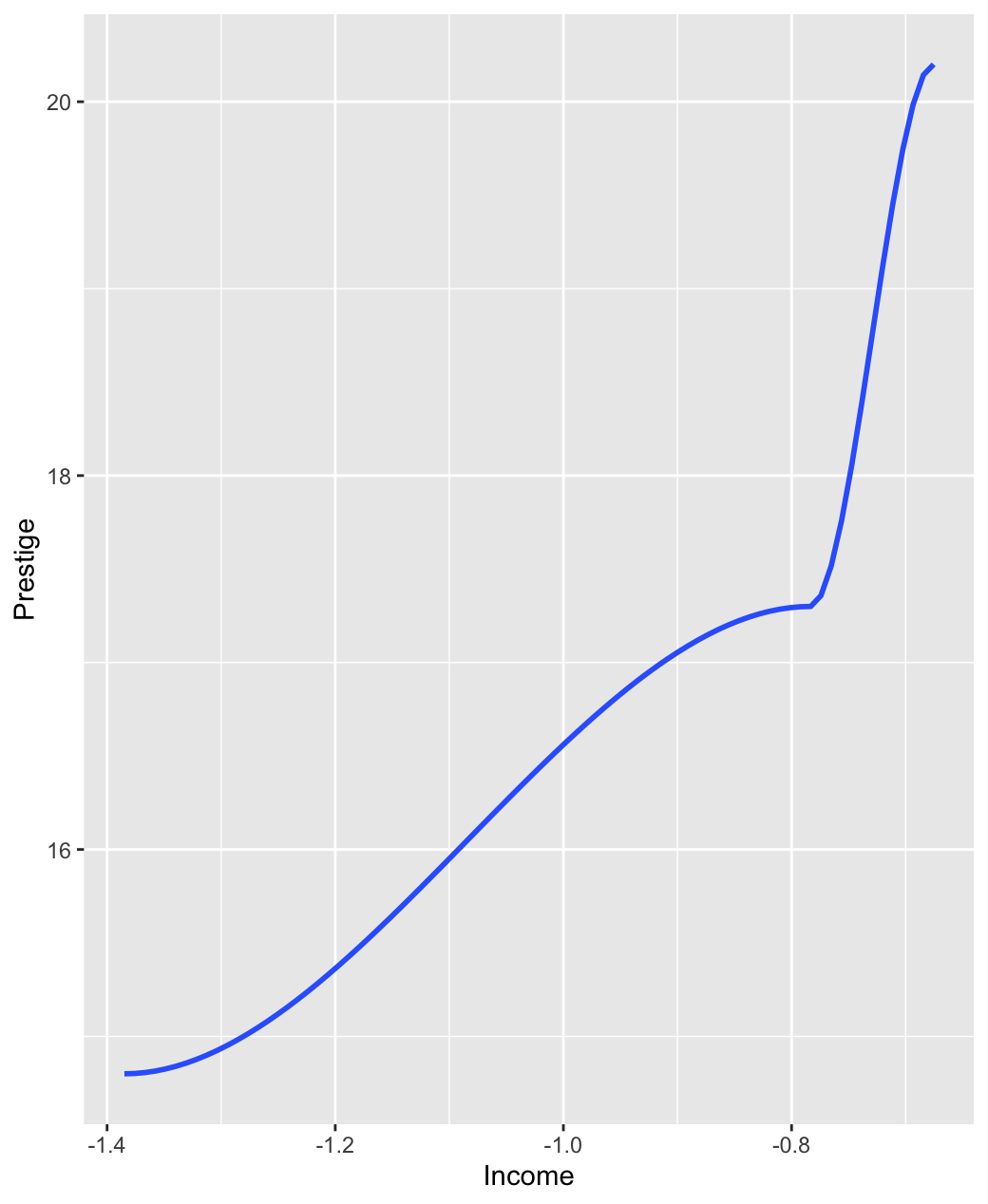}
}
\centering
\subfigure[Cluster 3]{
\includegraphics[width=0.3\linewidth, height=0.5\linewidth]{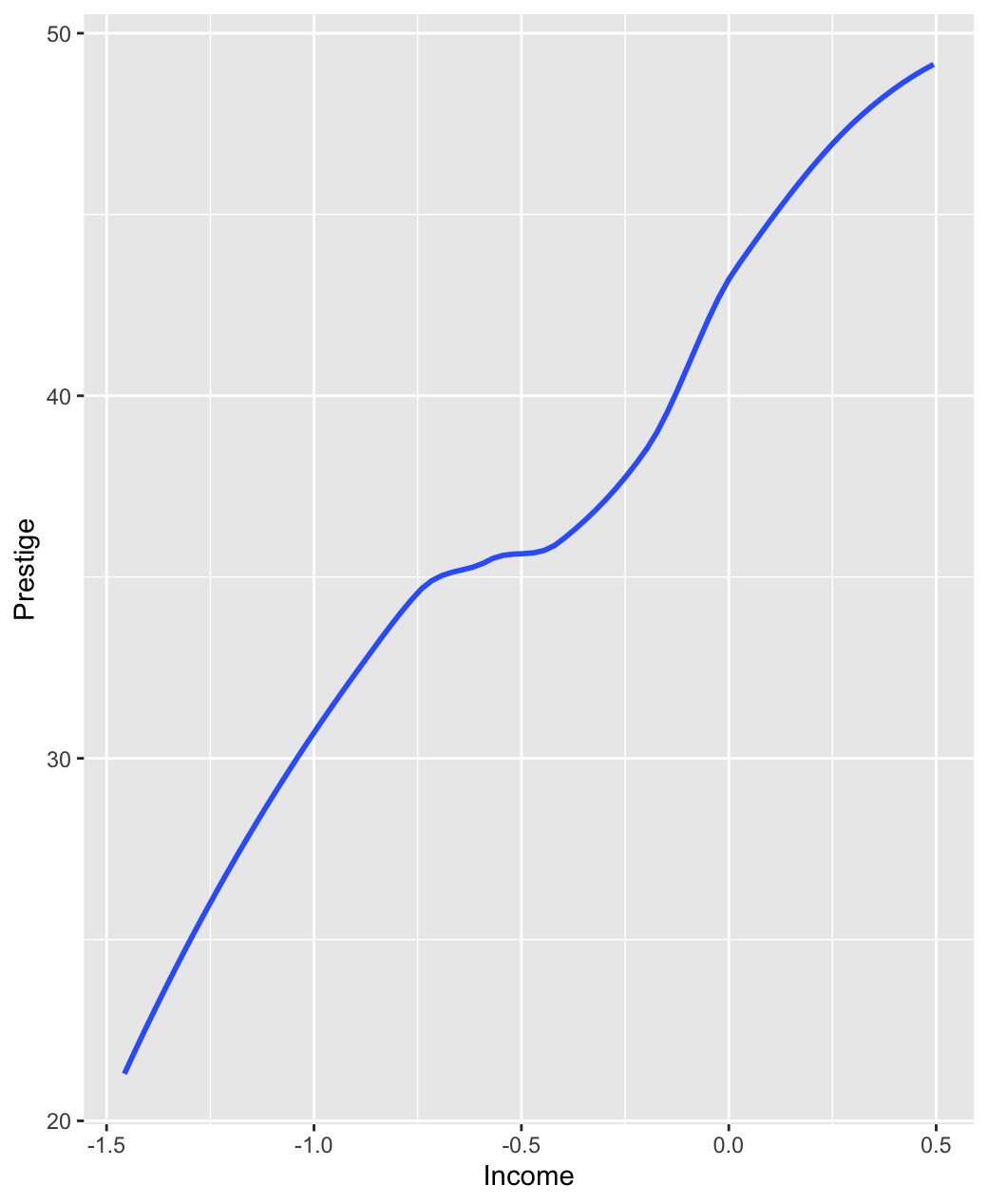}
}
\caption{Estimated $g_c(\cdot)$, $c=1,2,3$, through MoPLE for the Prestige dataset}\label{prestige}
\end{figure}

\subsection{Gross domestic product dataset}

In the second real data analysis, we examine gross domestic product (GDP) dataset sourced from the STARS database of World Bank. This dataset comprises information from 82 countries over the period 1960 to 1987 and includes some variables such as $\log$(GDP), indicating logarithm of real gross domestic product in million dollars, $\log$(Labor), representing logarithm of the economically active population aged 15 to 65, $\log$(Capital), implying logarithm of the estimated initial capital stock in each country, and $\log$(Education), denoting logarithm of the average years of education. 

Previously, researchers such as \cite{duffy2000cross} utilized this dataset to investigate the Cobb-Douglas specification, while \cite{wu2017estimation} examined how education and two other variables influence GDP using FMPLR with a fixed two-component mixture. 
In this paper, we investigate countries in 1975 with $Y = \log(\text{GDP})$, $\Xb = (\log(\text{Labor}), \log(\text{Capital}))$ and $U = \log(\text{Education})$, comparing clustering performance. To evaluate the clustering performance, we introduce a latent variable that indicates whether the country was classified as advanced or developing in 1975 based on International Monetary Fund (IMF).

\begin{table}[ht]
\caption{BIC values for each method in GDP dataset (Boldfaced numbers indicate the smallest value in each criterion)}
\label{tab:GDP_BIC} 
	\centering
	\begin{tabular}{ c | r r  r    } 
       \hline \hline\noalign{\smallskip}
		Number of clusters & MoE & FMPLR & MoPLE   \\ 	   \hline
        {1} & 74.46 & 337.10 & 337.10   \\
       {2} & 88.05 & 176.92 & $\boldsymbol{134.64}$   \\
      {3} & $\boldsymbol{60.95}$ & $\boldsymbol{169.90}$ & 178.15   \\
      {4} & 114.48 & 265.12 & 232.49  \\
      {5} & 110.04 & 419.94 & 405.89   \\
\noalign{\smallskip}\hline\noalign{\smallskip}
	\end{tabular}
\end{table}

\begin{table}[ht]
\caption{Clustering performance for each method in GDP dataset(Boldfaced numbers indicate the largest value in each criterion)}
\label{tab:GDP_clustering} 
	\centering
	\begin{tabular}{ c | r r  r    } 
      \hline  \hline\noalign{\smallskip}
		Index & MoE & FMPLR & MoPLE   \\ 	   \hline
        ARI &  0.3449 & -0.1238 & $\boldsymbol{0.7165}$    \\
       AMI & 0.3280 & 0.1042 & $\boldsymbol{0.6152}$   \\
  \noalign{\smallskip}\hline\noalign{\smallskip}
	\end{tabular}
\end{table}

Table \ref{tab:GDP_BIC} and Table \ref{tab:GDP_clustering} present the BIC values and clustering performance, respectively.
In Table \ref{tab:GDP_BIC}, MoPLE yield the expected number of clusters, while MoE and FMPLR selects more clusters than expected.
In Table \ref{tab:GDP_clustering}, MoPLE achieves the best results in terms of both ARI and AMI, followed by MoE. These findings suggest that MoPLE is the most suitable method when attempting to identify clusters among countries based on their classification as advanced or developing.

According to the results derived from MoPLE, the clusters labeled as $1$ and $2$ represent advanced and developing countries, respectively. {In addition, cluster $1$ reveals estimated coefficients for $\log(\text{Labor}  )$ and $\log(\text{Capital})$ as $(0.14, 0.86)$, while cluster $2$ displays coefficients as $(0.17, 0.82)$.} These results suggest that the impact of labor and capital on GDP does not significantly differ between advanced and developing countries.
Figure \ref{gdp} depicts the estimated $g_c(u)$ for each cluster, with $c = 1,2$. Specifically, in cluster $1$, the values of $\log$(GDP) appear to be higher compared to those in cluster $2$, while their shapes look similar.

\begin{figure}[] 
\centering
\subfigure[Cluster 1]{
\includegraphics[width=0.4\linewidth, height=0.5\linewidth]{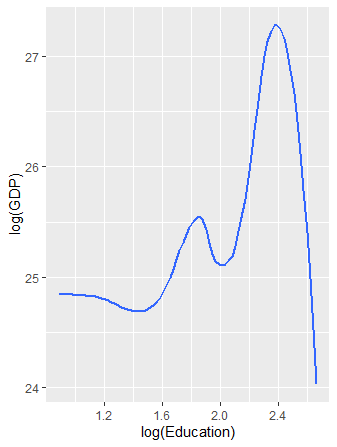}
}
\centering
\subfigure[Cluster 2]{
\includegraphics[width=0.4\linewidth, height=0.5\linewidth]{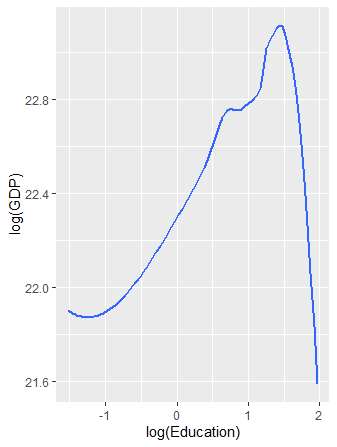}
}
\caption{Estimated $g_c(\cdot)$, $c=1,2$, through MoPLE for the GDP dataset}\label{gdp}
\end{figure}

\section{Discussion}

In this paper, we propose MoPLE, which applies a partial linear structrure to the expert network of MoE, replacing the linear structure. In numerical studies, MoPLE demonstrates the ability to estimate both parametric and non-parametric components effectively, not only under linear relationships between the response variable and covariates but also under non-linear relationships. Furthermore, it gives comparative performance in terms of the regression clustering. These results imply that MoPLE is a valuable model regardless of whether the data exhibits linear or non-linear relationships, excelling not only in parameter estimation but also in clustering.

While this study assumed univariate covariates for the non-parametric component, it is possible to extend this approach to higher dimensions. Nevertheless, we must acknowledge the curse of dimensionality as a limitation of non-parametric methods. One potential alternative approach is to structure each expert as a partially linear additive model.
Furthermore, although we postulate a specified variable following nonlinear relationships based on the previous work, it is still necessary to construct statistical hypothesis tests for nonlinear relationships, even though it may be challenging due to the presence of a hidden latent structure.

\bibliographystyle{apa} 
\bibliography{references}

\end{document}